\documentclass[12pt,twoside]{article}
\usepackage[mathscr]{eucal}
\usepackage{amsthm,amsmath,amssymb,amscd}
\usepackage{graphicx}
\usepackage{latexsym}
\newtheorem{theorem}{Theorem}

\theoremstyle{definition}

\def \beq{ \begin{equation} }
\def \eeq{\end{equation}}

\def \pd{\partial}

\def\.#1{\dot #1}

\def \( {\big( }
\def \) {\big) }

\def \bar{\overline}
\def \1{{\bf I}}
\def \2{{\bf II}}
\def\.#1{\dot #1}
\begin{document}
\DeclareGraphicsRule{.eps.gz}{eps}{.eps.bb}{#1}
\DeclareGraphicsRule{.ps.gz}{eps}{.ps.bb}{#1}
\title{\bf{A Counterexample to a Generalized Saari's Conjecture with a Continuum of Central Configurations}}

\author{{Manuele Santoprete}\renewcommand{\thefootnote}{\alph{footnote})}\footnotemark[1]\\Department of Mathematics\\294 Multipurpose Science \& Technology Bldg.\\
University of California, Irvine \\Irvine, California, 92697-3875 USA}

\renewcommand{\thefootnote}{\alph{footnote})}
\maketitle
\footnotetext[1]{Electronic mail: msantopr@math.uci.edu}

\begin{abstract}
\noindent In this paper we show that in the $n$-body problem with harmonic potential one can find a continuum of central configurations for $n=3$.
Moreover we show a counterexample to an  interpretation of Jerry Marsden  Generalized Saari's conjecture. This will help to refine our understanding and formulation of the Generalized Saari's conjecture, and in turn it might provide insight in how to solve the classical Saari's conjecture for $n\geq 4$.
\end{abstract}


\markboth{Santoprete M.}{A Counterexample to a Generalized Saari's Conjecture}

\section{Introduction}
Steve Smale proposed a set of problems for the 21st century. Smale's 6th problem  concerns the existence of a continuum of central configurations  in the $n$-body problem. In a nice paper Gareth Roberts found that there exists a continuum of central configurations in the $5$-body problem, where he allows some of the masses to have a negative value. Felipe Alfaro and Ernesto P\'erez-Chavela improved upon this result by finding a continuum of central configurations in the charged $4$-body problem. However the problem Steve Smale proposed is still open.

On the other hand, in 1970 Donald Saari made a beautiful conjecture:
{\it Every solution of the Newtonian $n$-body problem that has a constant moment of inertia is a relative equilibrium} \cite{Saari}.

Saari's conjecture has, in the last year, generated a good deal of interest.
In particular several partial result have been announced (see \cite{Diacu,Roberts2} for more details). Chris McCord provided a proof for the three body problem with equal masses \cite{McCord}. Jaume Llibre and Eduardo Pi\~na  found a different proof and an algorithm that could be used in the full three body problem \cite{Llibre}. Rick Moeckel has devised a computer-assisted proof for the full three body problem \cite{Moeckel}. Florin Diacu, Ernesto P\'erez-Chavela and Manuele Santoprete devised a proof for the collinear $n$-body problem \cite{Diacu}.  Nevertheless the conjecture remains open for $n\geq 4$.

Jerry Marsden  informally proposed a generalized Saari's conjecture at the Midwest Dynamical System Conference held at the University of Cincinnati, 4-7 October 2002. His conjecture concerns   mechanical systems with symmetry and gives new insight into the problem. The  original version of this conjecture states: {\it For a mechanical system with symmetry on the configuration manifold, the locked inertia tensor $I({\bf q})$ is constant along a solution if and only if ${\bf q}$ is a relative equilibrium}. The purpose  of  Jerry Marsden  was to  invite to find classes of dynamical systems with symmetry that verify the properties above rather than claim that all simple dynamical systems with symmetry satisfy the summentioned conditions.
This intent is more explicitly expressed in the ``Refined Saari problem" that improves and clarifies the statement above and was recently proposed by Antonio Hernand\'ez-Gardu\~no, Jeff Lawson and Jerry Marsden\footnote{The statement of the Refined Saari Problem appeared in \cite{Lawson} after this article was submitted for publication. All the remarks contined in this work apply equally well to the generalized Saari's conjecture and to the Refined Saari Problem.} (see \cite{Lawson}).

In this paper we wish to study an $n$-body problem where the particles interact by means of a harmonic potential.
We will show that a continuum of central configurations can be found in the $3$-body
problem with a special harmonic potential.
This result is quite interesting and complement  the ones found in the literature \cite{Roberts1,Perez}.
First of all, the examples in \cite{Roberts1,Perez} consider
interactions that are both attractive and repulsive, so that, in certain directions, the effect of different bodies cancels out.
In this paper we provide an example where there is no need to introduce repulsive forces and to cancel out the effect of some of the bodies.

Moreover, the example that we present  requires only $3$ bodies. This  the minimum number of bodies for which one can have a continuum of central configurations. Indeed for two bodies all the central configurations are equivalent, and therefore there can be only one central configuration.

In this paper we also provide a counterexample to an interpretation of the  Generalized Saari's conjecture above.
Some other examples concerning the $n$-body can be found in \cite{Chenciner, Chenciner1,Roberts2}. In \cite{Chenciner,Roberts2} the authors find counterexamples in the case of Hamiltonian systems with a inverse square law potential (the Jacobi, or ``pure Manev" potential). In \cite{Chenciner1, Roberts2} the autors present counterexamples for a class of homogeneous potentials with ``masses" of opposite sign.

The counterexample we present in this article is appealing for several reasons.
Firstly we do not need to introduce negative ``masses" and thus repulsive terms in the potential. Using negative masses changes the ellipsoids of constant inertia into  hyperboloids of constant inertia. This modifies the problem, since, in the simplest interpretation, the locked inertia tensor $I$ should be the moment of inertia i.e., a positive definite quadratic form.
Secondly our example is very simple and does not require a complicated analysis. Moreover it shows that  pathological behavior can be found even in the simplest of the $n$-body problems: the one with harmonic potential.
Furthermore it shows that the Jacobi potential is not an isolated case, but potentials with different power law provide a counterexample to the conjecture.

We would like to remark that one of the main benefits of finding counterexamples to the Generalized Saari's conjecture is to refine our  formulation and comprehension of  it, and this, in turn, will help us understand the classical Saari's conjecture and  provide insight in why and how the classical conjecture might fail.

Marsden's generalization has the merit of stimulating the development of new tools, such as techniques of geometrical mechanics, to study this problem. However, by now, it is clear that the generalized Saari's conjecture under discussion in this paper, as well as the Refined Saari Problem, do not include all simple mechanical systems with symmetry, but only some classes, as for example, as it is shown in \cite{Lawson}, system on Lie groups.

Let ${\bf q}=({\bf q}_1,{\bf q}_2...{\bf q}_n)\in {\mathbb R}^{2n}$ represent the positions of the $n$ bodies on the plane.
The Hamiltonian that describes the problem is of the form
\beq
H=\frac 1 2  \sum_1^n m_i\|\dot{\bf q}_i\|^2 +U({\bf q})
\eeq
and the equations of motions are
\beq
m_i\ddot{\bf q}_i=-\frac{\pd U}{\pd {\bf q}_i}
\eeq
A important quantity for this work is the moment of inertia
\beq
I({\bf q})=\sum_{i=1}^n m_i\|{\bf q}_i\|^2
\eeq
that can be written in terms of the mutual distances $r_{ij}=\|{\bf q}_i-{\bf q}_j\|$
as
\beq
I({\bf q})=\frac{1}{M}\sum_{i<j} m_im_jr_{ij}^2
\eeq
where $M=(m_1+m_2+...m_n)$.
A {\it central configuration} is a configuration $x\in {\mathbb R}^{2n}$ which satisfies the algebraic equation
\beq
\nabla I=\omega^2 \nabla U
\eeq
for some $\omega^2$. Therefore the central configurations are critical points of the potential energy $U$ restricted to the ellipsoids $I=k$.
Note that when counting central configurations it is standard to fix the size and identify configurations that are rotationally equivalent.

A given solution ${\bf q}={\bf q}(t)$ of the problem of $n$ bodies is called
{\it relative equilibrium} if there exists an orthogonal 2-matrix $\Omega=\Omega(t)$ such that for every $i$ and $t$ one has
\beq
{\bf q}_i= \Omega(t){\bf q}_i^0
\eeq
where ${\bf q}_i$, $\Omega$ belong to an arbitrary $t$ and ${\bf q}_i^0$ denotes ${\bf q}_i$ at some initial instant $t=t_0$
In such cases the system rotates about the
center of mass as a rigid body, the angular velocity is constant and the mutual
distances do not changes when $t$ varies.

In this paper we  consider a potential energy such that
\beq
U({\bf q})=\frac M 2 I({\bf q})
\label{pot}
\eeq
i.e. a particular case of coupled harmonic oscillators.
This potential is very peculiar because
\beq
\nabla U({\bf q})=\frac M 2\nabla I({\bf q})
\eeq
for every ${\bf q}$. This means that every point in configuration space is a central configuration. However in the case of two bodies all the configurations
are equivalent, and thus for an example of continuum of central configurations we need at least three bodies. An explicit example is given in the next section.


\begin{figure}[t!]
\begin{center}
\resizebox{!}{6.5cm}{\includegraphics{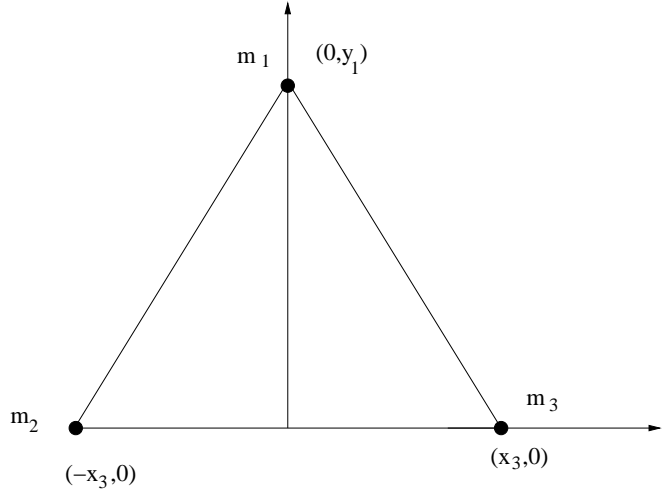}}
\end{center}
\caption{A symmetric configuration of three bodies.}
\label{triangle}
\end{figure}
\section{A Continuum of Central Configurations}
Let $P_1,P_2,P_3$ be three bodies of masses $m_1=m_2=m_3=1$, and let
${\bf q}_1=(0,y_1)$, ${\bf q}_2=(x_2,0)$ and ${\bf q}_3=(x_3,0)$ be the positions of the bodies where we take $x_2=-x_3$ (see Figure \ref{triangle}).
The moment of inertia, which can be written as
\beq
I=\frac 1 3(q_{12}^2+q^2_{13}+q^2_{23})=\frac 2 3 y_1^2+2x_3^2,
\eeq
defines an ellipse in the plane $(y_1,x_3)$.
Therefore the potential function restricted to the ellipsoid $I=k$ has a curve of critical points at
\beq
(y_1,x_3)=(\sqrt{\frac{3I}{2}}\cos\eta,\sqrt{\frac{I}{2}}\sin\eta)\quad \mbox{for }\quad 0\leq \eta\leq 2\pi.
\eeq
We have therefore proved the following

\begin{theorem}
In the three-body problem with harmonic potential given by (\ref{pot}), there exists a one-parameter family of degenerate central configurations where the three equal masses are positioned at the vertices of a isosceles triangle.
\end{theorem}

We wish to remark that for $\eta=0,\pi,2\pi$ the bodies $P_2$ and $P_3$ occupy the same position.

\section{A Counterexample to the Generalized Saari's Conjecture}
We begin our description of the counterexample to the generalized Saari's conjecture
considering four bodies $P_1,P_2,P_3,P_4$ at the vertices of a rhombus (see Figure \ref{rhombus}).
We let $m_1=m_2=m_3=m_4=1$ be the four masses and ${\bf q}_1=(0,y_1)$, ${\bf q}_2=(x_2,0)$, ${\bf q}_3=(x_3,0)$, ${\bf q}_4=(0,y_4)$, be the positions of the bodies, where $x_2=-x_3$ and $y_4=-y_1$.
\begin{figure}[t!]
\begin{center}
\resizebox{!}{6.5cm}{\includegraphics{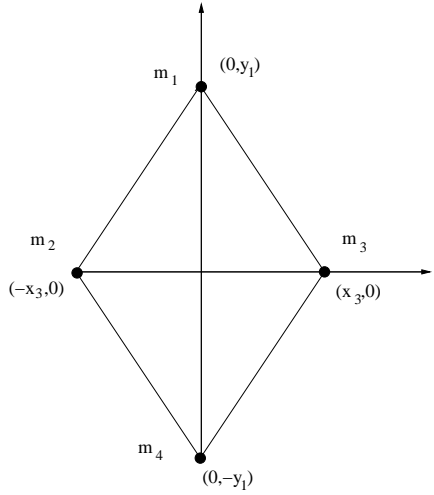}}
\end{center}
\caption{A symmetric configuration of four bodies.}
\label{rhombus}
\end{figure}
We want to consider only those solutions the configuration of which is a rhombus at all times.
The potential energy in this case is
\beq
U=\frac 1 2(r_{12}^2+r_{13}^2+r_{14}^2+r_{23}^2+r_{24}^2+r_{34}^2).
\eeq
 We can also rewrite the potential energy in terms of the coordinates of the bodies
\beq
U=\frac 1 2(2x_2^2+2x_3^2+2y_1^2+2y_4^2+(y_1-y_4)^2+(x_3-x_2)^2).
\eeq
The equations of motion for the bodies $P_1$ and $P_3$ are
\beq
\begin{split}
\ddot y_1=&-\frac{\pd U}{\pd y_1}=-2y_1-(y_1-y_4)\\
\ddot x_3=&-\frac{\pd U}{\pd x_3}=-2x_3-(x_3-x_2).
\end{split}
\label{eqbodies}
\eeq
and the equation of the other two bodies can be trivially deduced from them.
Since  $x_2=-x_3$ and $y_4=-y_1$, and the moment of inertia is such that $U=(M/2)I$, it is easy to see that
\beq
I=2(x_3^2+y_1^2)=U/2.
\eeq
Thus, if we let $I=k$, then the potential energy is also constant, and the curves $I=k$
are circles in the plane $(y_1,x_3)$.
Because of the high symmetry of the spatial configuration  of the bodies (i.e. $x_2=-x_3$ and $y_4=-y_1$)  the equation of motion (\ref{eqbodies}) for  $P_1$ and $P_3$ can be written as
\beq
\begin{split}
\ddot y_1=&-\frac{\pd U}{\pd y_1}=-4y_1\\
\ddot x_3=&-\frac{\pd U}{\pd x_3}=-4x_3.
\end{split}
\eeq
A particular solution of the above equations, with initial conditions $\bar y_1(0)=(\sqrt{I/2},0)$ and $\bar x_3(0)=(0,\sqrt{I/2})$, is given by the following equations

\beq
\begin{split}
\bar y_1=&\sqrt{\frac{I}{2}}\cos(2t)\\
\bar x_3=&\sqrt{\frac{I}{2}}\sin(2t).
\label{oscillators}
\end{split}
\eeq
Clearly the moment on inertia along the solution above is
\beq
2(\bar y_1^2+\bar x_3^2)=I
\eeq
that is constant by hypothesis. The solution of the differential equations (\ref{oscillators}) can be viewed as parametric equations of the circle.
Furthermore the solution (\ref{oscillators}) is  not a relative equilibrium according to the definition given at the beginning of the paper. More precisely we have the following
\begin{theorem}
The trajectory defined by
\[(\bar y_1,\bar x_3)=\left(\sqrt{\frac{I}{2}}\cos(2t),\sqrt{\frac{{I}}{2}}\sin(2t)\right)
\]
is a solution of the four-body problem with harmonic potential given by equation (\ref{pot}). This solution has constant moment of inertia $I=k$ and is not a relative equilibrium.
\end{theorem}
\begin{proof}
To prove the theorem we only have to show that the solution found above is not a relative equilibrium.
We can focus our attention on the body $P_1$. The orbit of the body $P_1$ along  the solution described above is ${\bf q}_1=(0,\sqrt{I/2}\cos(2t))$ and its position  at the instant $t=t_0=0$
is ${\bf q}^0_1=(0,\sqrt{I/2})$. Therefore we can write ${\bf q}_1=\Omega(t){\bf q}^0_1$ where
\beq \Omega(t)=\left (
\begin{array}{cc}
A(t)&0\\
B(t)&\cos(2t)
\end{array}\right ).
\eeq
However the matrix $\Omega(t)$ cannot be orthogonal for every $t$.  $\Omega(t)$ is orthogonal for any $t$ if and only if it  satisfies the equation $\Omega(t)\Omega^t(t)=Id$ for every $t$, where $Id$ is the two by two identity matrix. Explicitely we have
\beq\Omega(t)\Omega(t)^t=\left (
\begin{array}{cc}
A(t)^2& A(t)B(t)\\
A(t)B(t)& B(t)^2+\cos^2 (2t)
\end{array}\right)=
\left(
\begin{array}{cc}
1&0\\
0&1
\end{array}\right)
\eeq
from which we obtain that $A(t)^2=1$ and $B(t)^2=\sin^2(2t)$. But this implies that $A(t)B(t)\neq 0$  for some values of $t$. For example it is trivial to verify that $A(\pi/4)B(\pi/4)\neq 0$. This shows that $\Omega(t)$ is not orthogonal for  $t=\pi/4$ and thus the solution studied above is not a relative equilibrium. 
\end{proof}
Observe that for $t=k\pi$ the bodies $P_2$ and $P_3$ collide, while for  $t=\pi/2+k\pi$ the bodies $P_1$ and $P_2$ collide. However the equation of motions do not encounter singularities. Two different interpretations can be given to the equations, either one assume that the particles go through each other or that they have an elastic collision. We allow the particles to go trough each other. But, simply changing the notation, it is easy to introduce elastic collisions into the problem.

\section*{Acknowledgments}
The  author would like to thank Donald Saari and Ernesto P\'erez-Chavela for their advise and suggestions regarding this work.

\end{document}